\documentclass[conference,letterpaper]{IEEEtran}

\IEEEoverridecommandlockouts

\usepackage{amsthm}
\usepackage{amssymb}
\usepackage{amsmath}
\usepackage{amssymb}
\usepackage{epsfig}
\usepackage{epsf}
\usepackage{subfigure}
\usepackage{graphicx}
\usepackage{cite}
\usepackage{url}
\usepackage[]{authblk}
\usepackage{color}
\usepackage{upgreek}
\usepackage{comment}

\usepackage{wrapfig}

\usepackage{pifont}

\def\BibTeX{{\rm B\kern-.05em{\sc i\kern-.025em b}\kern-.08em
    T\kern-.1667em\lower.7ex\hbox{E}\kern-.125emX}}

\newtheorem{theorem}{Theorem}

\newtheorem{lemma}{Lemma}

\newcommand{\msf}{\mathsf}

\newcommand{\lp}{\left(}
\newcommand{\rp}{\right)}

\newcommand{\ra}{\mathrm{RIS}_1}
\newcommand{\rb}{\mathrm{RIS}_2}

\newcommand{\dn}{\delta_{N}}
\newcommand{\db}{\dn}
\newcommand{\ds}{\delta_{S}}
\newcommand{\dm}{\ds}
\newcommand{\dd}{\delta_{D}}
\newcommand{\da}{\dd}
\newcommand{\dbar}{\bar{\delta}}

\newcommand{\si}{\mathcal{A}^n}

\begin{document}




\title{Capacity-Maximizing Dynamic User Association in Double RIS-Aided Broadcast Networks}

\author{
Alireza~Vahid,~\IEEEmembership{Senior~Member,~IEEE}
\thanks{Alireza Vahid is with the Electrical and Microelectronic Engineering Department at Rochester Institute of Technology, Rochester, NY 14623, USA. Email: {alireza.vahid@rit.edu}.} 
}

\maketitle


\begin{abstract}
We introduce an information-theoretic framework to dynamically pair up different reconfigurable intelligent surfaces (RISs) with wireless users with goal of maximizing the fundamental network capacity. We focus on a double RIS-aided broadcast packet network with two users. We show using a dynamic RIS-user association and an opportunistic protocol, the network capacity could be significantly enhanced and superior to other benchmarks with static associations. The results include new outer-bounds on network capacity and their achievability. We discuss the optimal RIS-user association.
\end{abstract}



\section{Introduction}
\label{Section:Introduction}


In recent years, new wireless technologies such as reconfigurable intelligent surfaces (RISs) have emerged as potential enablers of 6G and future wireless generations~\cite{liu2021reconfigurable}. 
In simple terms, an RIS can be viewed as a set of programmable reflecting elements that can alter the phase and potentially the amplitude of the incoming signals.
The key strengths of an RIS are its ability to manipulate the wireless medium, full-duplex, and low power consumption.
The ability to manipulate the wireless medium is particularly important in sustaining connectivity in higher frequency bands, which are of interest in vehicular communications, where electromagnetic waves suffer from high path-loss and may be easily blocked~\cite{peng2018statistical,DebbahTHz22}.
Collectively, these aspects make RIS an attractive solution in vehicular networks~\cite{al2022reconfigurable,ai2021secure,agrawal2021performance,hoang2023secrecy,mensi2022performance,hoang2024physical}, which are typically in dense urban environments and are characterized by high mobility and high quality of service requirements.
However, RIS-aided networks face various challenges such as the overhead of channel learning and the computational costs and delays associated with the optimization of the RISs.
Moreover, in general, the fundamental Shannon capacity of RIS-aided networks, like most multi-user communication networks, remain unknown.
Further, and from a MAC-layer perspective, there needs to be an RIS-user association in place based on which communication algorithms could be designed and implemented.


There are quite a few results on RIS-aided communications. 
A large body of the literature assumes the channel knowledge is available~\cite{bafghi2022degrees,gan2021user}, which is hard if not impossible to achieve in large-scale and mobile networks.
Other results assume statistical~\cite{gan2021ris} or noisy/local~\cite{nassirpour2023beamforming} knowledge of the channels.
Then, many results focus on eliminating dead zones or maximizing signal-to-noise ratio (SINR) at the users~\cite{kammoun2020asymptotic,al2021ris} instead of considering the Shannon capacity due to the associated complexities. 
We note that SINR maximization may be far from the optimal Shannon capacity~\cite{etkin2008gaussian}.
RIS-user association has also been studied, based on techniques such as the nearest-neighbor, but not in the context of maximizing Shannon capacity~\cite{gan2021user,nassirpour2023beamforming,liu2022joint}.
Finally, the prior RIS-user associations are typically static meaning that they remain unchanged for the course of communications.


To overcome some of these challenges and provide an information-theoretically optimal answer to RIS-user association, we consider a broadcast packet network where the communication channels are modelled as erasure links.
We note that packet networks are well-suited to model ``last-mile communications'' and erasure links capture the frequent interruptions in higher frequency bands.
These observations, make our study relevant to vehicular networks in dense environments with high mobility.
More specifically, we consider a double RIS-aided broadcast packet network with one transmitter and two receiver terminals.
We allow the two RISs to steer their beams toward different users multiple times throughout the communication block, which we refer to as dynamic RIS-user association.
For the two-user double RIS-aided broadcast packet network, our contributions are summarized as below:
\begin{itemize}
    \item We provide an information-theoretic RIS-user association, which allows for dynamic changes during communications. We show the resulting network capacity is superior to static RIS-user associations;
    \item We provide new bounds the define the boundaries of the maximum attainable rates of reliable communication;
    \item We devise a linear protocol that leverages the statistical changes induced by the RIS-user association to maximize the spectrum efficiency and achieve the bounds, thus, characterizing the network capacity;
    \item Instead of relying on exact channel knowledge, we rely on short-length causal ACK/NACK signaling incorporated in most packet networks;
    \item We provide detailed examples to clarify the ideas and establish the superiority of the proposed scheme over static RIS-user associations.
\end{itemize}

The paper is organized as follows. Section~\ref{Section:Setup} describes the network and the required definitions. Section~\ref{Section:Main} presents the new information-theoretic bounds and Section~\ref{Section:Example} includes several examples to compare the findings with other benchmarks and demonstrate the potential gain of our dynamic RIS-user association. Section~\ref{Section:Proofs} is dedicated to the proof of the bounds. Section~\ref{Section:Conclusion} concludes the paper.


\begin{figure*}[!ht]
\centering
\includegraphics[width = .8\textwidth]{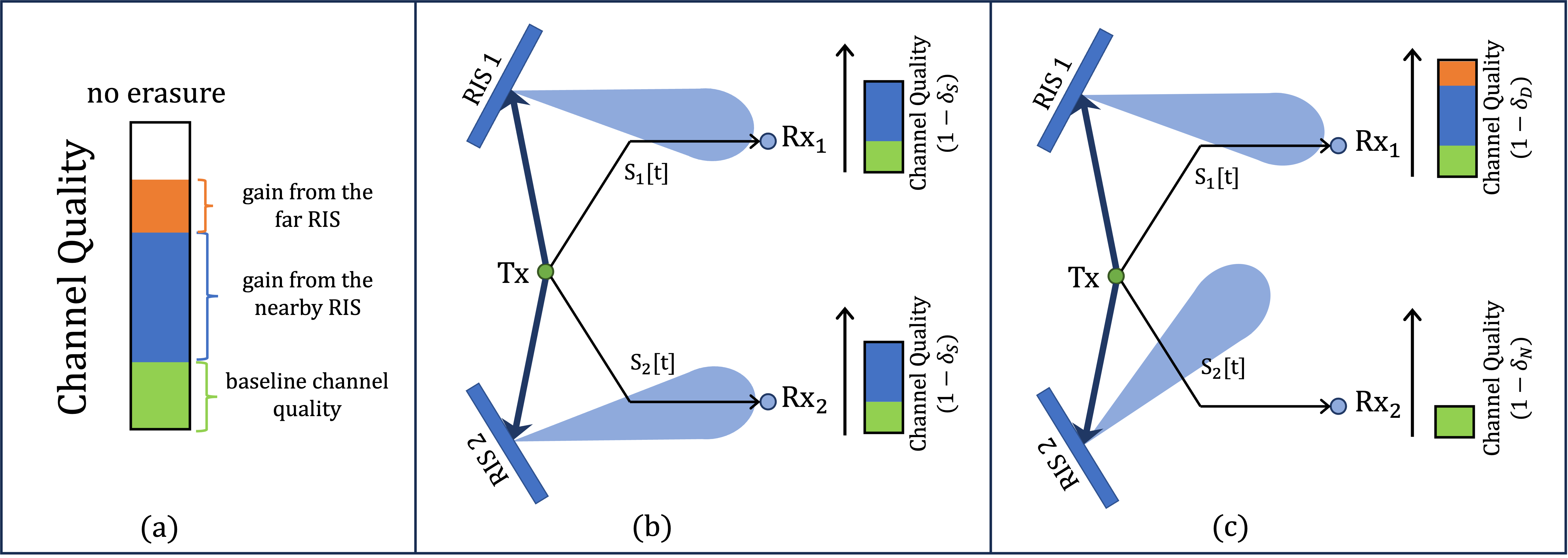}
\caption{The two-user double RIS-aided broadcast packet network: (a) impact of RIS association on channel quality measured in terms of erasure likelihood; (b) $\mathrm{RIS}_i$ assists the reception at $\msf{Rx}_i$, $i=1,2$; (c) both RISs are associated with $\msf{Rx}_1$ resulting in the lowest erasure probability (highest channel quality).\label{Fig:AssignIT}}
\end{figure*}

\section{System Setup}
\label{Section:Setup}

To better understand how RIS-user association affects network capacity and the corresponding transmission protocol, we focus on one of the key building blocks in network information theory and wireless communications, namely the two-user broadcast packet network, where the communication is enhanced by two reconfigurable intelligent surfaces. 
More specifically, the network consists of five distinct nodes: a transmitter node, $\msf{Tx}$, two receiver terminals, $\msf{Rx}_1$ and $\msf{Rx}_2$, and two reconfigurable intelligent surfaces, $\ra$ and $\rb$.

We assume the transmitter, $\msf{Tx}$, has two independent files, $W_1$ and $W_2$, and wishes to reliably deliver them to the two receivers over $n$ channel uses, where each channel use is sufficient for the communication of one packet.
Each message, $W_i$, contains $|W_i|=m_i=nR_i$ data packets in $\mathbb{F}_q$ where $R_i$ is the rate for user $i$ in packets per channel use, $i=1,2$.
In general, the transmitter and receivers may have multiple antennas and the RISs will consists of several reflecting elements. 
These aspects of the communication link from the transmitter to receiver terminal $\msf{Rx}_i$ at time $t$ are modelled in this work by an erasure link $S_i[t] \in \{ 0,1\}$, whose erasure probability is determined by the RIS-user association at time $t$.
In other words, at each time, the transmitted packet is either delivered successfully to a given terminal or erased, and the erasure likelihood is determined by the choice of RIS-user association during the communication of that particular packet. 
Below, we mathematically formulate our model.

\noindent \underline{\bf RIS-user association:} 
An RIS-user association over a communication block of length $n$ is denoted by: 
\begin{align}
    \mathcal{A}^n = \left\{ \mathcal{A}[1], \mathcal{A}[2], \ldots, \mathcal{A}[n] \right\},
\end{align}
where $\mathcal{A}[t]$ is itself a set that defines the RIS-user association at time instant $t = 1,2,\ldots, n$.
More specifically,
\begin{align}
    \label{Eq:Assignt}
    \mathcal{A}[t] = \left\{ (\msf{Rx}_1 \leftarrow \mathcal{B}_1[t] ), (\msf{Rx}_2 \leftarrow \mathcal{B}_2[t] ) \right\},
\end{align}
where $\mathcal{B}_1[t], \mathcal{B}_2[t] \subseteq \{ \ra, \rb \}$, $\mathcal{B}_2[t] = \mathcal{B}_1^\complement[t]$ ($\cdot^\complement$ is the set complement operation).
In other words, \eqref{Eq:Assignt} means that at time instant $t$, the communication link from the transmitter to $\msf{Rx}_1$ is enhanced by the RIS(s) in $\mathcal{B}_1[t]$, and the communication link from the transmitter to $\msf{Rx}_2$ is enhanced by RIS(s) in $\mathcal{B}_2[t] =\mathcal{B}_1^\complement[t]$.
We note that $\mathcal{B}_i[t]$ may be $\emptyset$.
Moreover, while in general an RIS may create multiple beams, from a theoretical standpoint, that can be treated as having additional RISs. 
Thus, here we only focus on the case where each RIS may be associated with at most one user.

\noindent \underline{\bf Channel statistics:} 
As mentioned above, the channel from the transmitter, $\msf{Tx}$, to receiver terminal $\msf{Rx}_i$ is described by the erasure link $S_i[t] \in \{ 0,1\}$, $t = 1,2,\ldots, n,$ and $i =1,2$. 
We assume the erasures occur independently across users and time, but the likelihood is determined by the RIS-user allocation at time instant $t$.
We further assume $\ra$ is positioned closer to $\msf{Rx}_1$, and $\rb$ is positioned closer to $\msf{Rx}_2$. 
Then, if $\mathcal{B}_1[t]$ is equal to:
\begin{enumerate}
    \item $\emptyset$, then $S_1[t] \sim \mathcal{B}(1-\dn)$;
    \item $\{ \ra \}$, then $S_1[t] \sim \mathcal{B}(1-\ds)$;
    \item $\{ \ra, \rb \}$, then $S_1[t] \sim \mathcal{B}(1-\dd)$,
\end{enumerate}
where $0 \leq \dd < \ds < \dn \leq 1$.
We note that due to the position of the RISs, we ignored $\mathcal{B}_1[t] = \{ \rb \}$ as its inclusion would not improve or change the findings. 
A similar set of statements holds for $\msf{Rx}_2$ and $\mathcal{B}_2[t]$.
Figure~\ref{Fig:AssignIT}(a) illustrates an example of the potential improvements in channel quality (\emph{i.e.}, decrease in erasure probability) as the user gets more aid from the RISs.
Intuitively, the gain from the nearby RIS should be larger, which is one reason why we ignored associations such as $\mathcal{B}_1[t] = \{ \rb \}$.
Figures~\ref{Fig:AssignIT}(b) and (c) showcase the network with some possible RIS-user associations and the resulting channel quality at each terminal.

\noindent \underline{\bf Received signals:} At time instant $t$, the files are mapped to channel input $X[t] \in \mathbb{F}_q$, and the corresponding received signals are:
\begin{align}
\label{eq_DL_channel}
Y_1[t] = S_1[t] X[t]~~ \; \mbox{and} \;~~ Y_2[t] = S_2[t] X[t],
\end{align}
where $S_i[t] = 0$ corresponds to the erasure of the transmitted packet at $\msf{Rx}_i$, $i =1,2$. 

\noindent \underline{\bf Other assumptions:} 
We assume the receivers will decode their desired files at the end of the communication block. 
At every time instant $t$, each receiver estimates the channel strength, and will inform other nodes via short-length ACK/NACK signaling whether the packet was erased or not. 
We further assume $\mathcal{A}^n$ is determined a priori and known to all nodes.

\noindent \underline{\bf Channel input:} The constraint imposed at time index $t$ on the encoding function $f_t(.)$ at the transmitter is:
\begin{align}
\label{eq_enc_function}
X[t] = f_t\lp W_1, W_2, S^{t-1}, \si \rp,
\end{align}
where $S^{t-1}=(S_1^{t-1}, S_2^{t-1})$, which captures the causal ACK/NACK signaling from both receiver terminals up to time instant $(t-1)$. 
 
\noindent \underline{\bf Decoding:} 
Receiver $\msf{Rx}_i$, $i=1,2$, will decode its desired file at the end of the communication block using the decoding function $\varphi_{i,n}\left( Y_i^n, S^n, \si \right)$. An error occurs whenever $\widehat{W}_i \neq W_i$. The average probability of error is given by
\begin{align}
\lambda_{i,n} = \mathbb{E}[P(\widehat{W}_i \neq W_i)],
\end{align}
where the expectation is taken with respect to the random choice of the transmitted files.

\noindent \underline{\bf Network capacity:} 
For a given RIS-user association, $\si$, we say that a rate-pair $(R_1,R_2)$ is achievable, if there exist a block encoder at the transmitter and a block decoder at each receiver, such that $\lambda_{i,n}$ goes to zero as the block length $n$ goes to infinity. 
The capacity region, $\mathcal{C}_\mathcal{A}$, is then the closure of the set of the achievable rate-pairs.

\section{Information-Theoretic Bounds}
\label{Section:Main}

Fix $0 < \eta_1, \eta_2 < 1/2$ such that $\eta_1 n, \eta_2 n \in \mathbb{Z}^+$. Based on this, we define the following RIS-user association:
\begin{itemize}
    \item $\mathcal{B}_1[t] = \{ \ra, \rb \}$ for $t = 1, \ldots, \eta_1 n$;
    \item $\mathcal{B}_1[t] = \emptyset$ for $t = \eta_1 n + 1, \ldots, (\eta_1+\eta_2) n$;
    \item $\mathcal{B}_1[t] = \{ \ra \}$ for $t = (\eta_1+\eta_2) n + 1, \ldots, n$.
\end{itemize}
We refer to this RIS-user association as the ``dynamic association,'' and if $\eta_1 = \eta_2$, we refer to it as the symmetric dynamic association. The following theorem establishes a set of outer-bounds on the maximum attainable rate of reliable communications for this dynamic association.

\begin{theorem}
\label{THM:RIS-BEC}
For the two-user double RIS-aided broadcast packet network described in Section~\ref{Section:Setup} and for the dynamic RIS-user association defined above, we have:
\begin{equation}
\label{Eq:Theorem}
\mathcal{C}_{\mathcal{A}} \subseteq \mathcal{C}^{\msf{out}} \equiv \left\{ \begin{array}{ll}
0 \leq R_1 + \beta R_2 \leq \beta(1-\dbar_2), & \\
0 \leq \beta R_1 + R_2 \leq \beta(1-\dbar_1), & 
\end{array} \right.
\end{equation}
where 
\begin{equation}
    \label{Eq:beta}
    \beta = \left\{ \begin{array}{ll}
    \frac{1-\da\dn}{1-\dn}, & \msf{~if~} \ds \leq \frac{\dn(1-\dd)}{1-\dn} \\
    1+\ds, & \msf{~if~} \ds > \frac{\dn(1-\dd)}{1-\dn} 
    \end{array} \right.
\end{equation}
and
\begin{align}    
    \label{Eq:dbar}
    \dbar_1 &= \eta_1 \dd + \eta_2 \dn + (1-\eta_1-\eta_2) \ds, \nonumber\\
    \dbar_2 &= \eta_1 \dn + \eta_2 \dd + (1-\eta_1-\eta_2) \ds.
\end{align}
\end{theorem}

To better understand the implications of the bounds presented in Theorem~\ref{THM:RIS-BEC}, to see whether these bounds can be achieved in a practical setting, and to compare the results with static association benchmarks, the following section includes a set of detailed examples. 
The proof for this theorem is presented in Section~\ref{Section:Proofs}.

\section{Motivation, Analysis, and Comparisons}
\label{Section:Example}

In this section, we present a sample network, which we will use to illustrate how a dynamic RIS-user association may significantly improve network capacity compared to a static association. 
In particular and as our benchmarks, we consider a network with no RIS, and two static RIS-user associations in the double RIS-aided network.
For this section and for illustration purposes, we fix $\dn = 0.8$ (no RIS aid: weak channel/high erasure), $\ds = 0.5$ (single RIS aid: moderate channel/moderate erasure), and $\dd = 0.3$ (double RIS aid: strong channel/low erasure).

\noindent \underline{\bf No RIS:} If there is no RIS in the network, the channel from the transmitter to each receiver terminal would be governed by a $\mathcal{B}(1-\db$) process. 
In other words, for the entire communication block, the erasure probabilities would be the same.
Such a homogeneous problem with ACK/NACK has been studied before~\cite{GatzianasGeorgiadis_13} and its capacity region is known:
\begin{equation}
\label{Eq:InactiveRIS}
\mathcal{C}_{\msf{noRIS}} \equiv \left\{ \begin{array}{ll}
0 \leq  R_1 + (1+\db) R_2 \leq ( 1 - \db^2), & \\
0 \leq (1+\db) R_1 + R_2 \leq (1-\db^2). & 
\end{array} \right.
\end{equation}

Here, we present a brief outline of the transmission protocol from~\cite{GatzianasGeorgiadis_13} that achieves $\mathcal{C}_{\msf{noRIS}}$, which is shown in Figure~\ref{Fig:PlanComparison} for the parameters given above.
From the figure, it is evident that we only need to describe the achievability of the symmetric sum-rate point and the rest of the region can be achieved by time-sharing.
The protocol has three phases. 
In phase $1$ (phase $2$), packets intended for $\msf{Rx}_1$ ($\msf{Rx}_2$) are transmitted until {\it at least one} receiver terminal obtains each packet. 
Then, the transmitter creates new packets by creating pairwise summations of the packets that are intended for $\msf{Rx}_1$ but available at $\msf{Rx}_2$ with those intended for $\msf{Rx}_2$ but available at $\msf{Rx}_1$.
Note that the resulting new packets would have the exact same size as the original packets but each bit is the results a finite field summation of the corresponding bits in the two original packets.
These new packets are beneficial to both receivers as one part is always known and the other is desired.
The transmitter will deliver these new packets to {\it both} receiver terminals in phase $3$.

\begin{figure}[!t]
\centering
\includegraphics[width = .7\columnwidth]{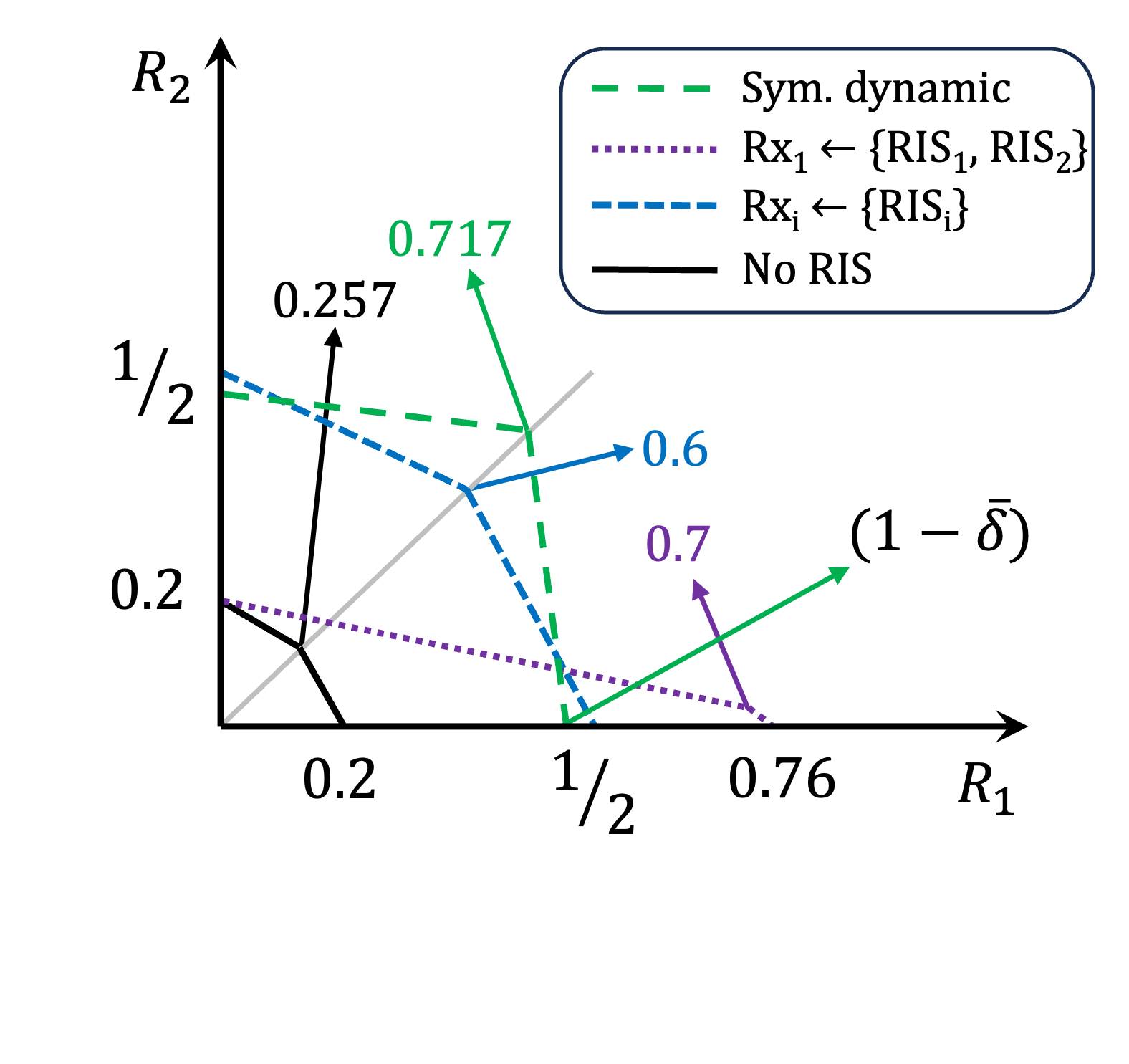}
\vspace{-6mm}
\caption{Comparing the achievable rates for different RIS-user associations.\label{Fig:PlanComparison}}
\end{figure}

\noindent \underline{\bf Single RIS per user:} If for the entire communication block, $\mathrm{RIS}_i$ assists the reception at $\msf{Rx}_i$, $i=1,2$, as in Figure~\ref{Fig:AssignIT}(b), then the channels from the transmitter to the receiver terminals would be governed by $\mathcal{B}(1-\ds$) processes. Similar to the previous case:
\begin{equation}
\label{Eq:Neutral}
\mathcal{C}_{\msf{neutral}} \equiv \left\{ \begin{array}{ll}
0 \leq  R_1 + (1+\dm) R_2 \leq ( 1 - \dm^2), & \\
0 \leq (1+\dm) R_1 + R_2 \leq (1-\dm^2). & 
\end{array} \right.
\end{equation}

\noindent \underline{\bf Double RIS aide for one user:} If instead, both RISs are associated with $\msf{Rx}_1$ for the entire block, \emph{i.e.}, $\mathcal{A}[t] = \left\{ (\msf{Rx}_1 \leftarrow \{ \ra, \rb \} ), (\msf{Rx}_2 \leftarrow \emptyset ) \right\}$ for $t=1,2,\ldots,n$, then we have the following capacity region:
\begin{equation}
\label{Eq:OneUser}
\mathcal{C}_{\msf{user~1}} \equiv \left\{ \begin{array}{ll}
0 \leq R_1 + \frac{1-\da\db}{1-\db} R_2 \leq (1-\da\db), & \\
0 \leq  \frac{1-\da\db}{1-\da} R_1 + R_2 \leq (1-\da\db). & 
\end{array} \right.
\end{equation}

\noindent \underline{\bf Symmetric dynamic association:} The transmitter may instead decide to a dynamic association. Fix $\eta_1 = \eta_2 = \eta$. Then, $S_1[t] \sim \mathcal{B}(1-\da)$ and $S_2[t] \sim \mathcal{B}(1-\db)$ for $t=1,2,\ldots, \eta n$; during the following $\eta n$ channel uses,  $\msf{Rx}_2$ will have a better channel quality; and during the remaining $(1 - 2\eta) n$ channel uses, both channels are governed by $\mathcal{B}(1-\ds$) processes.

\noindent \underline{\bf Transmission protocol:} 
Our opportunistic protocol is a modification of the one presented in~\cite{GatzianasGeorgiadis_13}, which we outlined above.
Fix $\eta = 1/2.12$. 
The goal is to successfully deliver $m = 0.76/2.12 n$ packets to each receiver terminal, which would result in a sum-rate of $2m/n = 0.76/1.06 \approx 0.717$. 
The communication block is divided into three phases.

\noindent {\bf Transmission in phase $1$:} At each time instant, the transmitter will send out a packet intended for $\msf{Rx}_1$ until at least one terminal obtains it. If the packet arrives at $\msf{Rx}_1$, no further action is required. If the packet is delivered $\msf{Rx}_2$ but not $\msf{Rx}_1$, the transmitter will keep this packet in a virtual queue. The communication will last an average of\footnote{Average values are used for convenience. A careful analysis using concentration theorems~\cite{vahid2021erasure} would ensure the asymptotic rate when $n \rightarrow \infty$, matches the one obtained by using the average values.}:
\begin{align}
\frac{1}{1-\da\db}m = \eta n.
\end{align}
Further, there will be on average:
\begin{align}
\frac{\da(1-\db)}{1-\da\db}m 
\end{align}
packets in the virtual queue.

\noindent {\bf Transmission in phase $2$:} Similar to phase $1$ but with interchanging IDs.

\noindent {\bf Transmission in phase $3$:} The transmitter will first create the pairwise summations of the packets in the two virtual queues. 
Then, the resulting sequence will be sent to the terminals at an erasure code of rate $(1-\dm)$. The transmission in this mode will on average take:
\begin{align}
\frac{\da(1-\db)}{(1-\dm)(1-\da\db)}m = (1-2\eta)n.
\end{align}

\noindent {\bf Achievable region:} The transmission described above achieves the following region: 
\begin{equation}
\label{Eq:AchBalanced}
\mathcal{R}_{\msf{dynamic}} \equiv \left\{ \begin{array}{ll}
0 \leq R_1 + \frac{1-\da\db}{1-\db} R_2 \leq \frac{1-\da\db}{1-\db}(1-\dbar), & \\
0 \leq  \frac{1-\da\db}{1-\db} R_1 + R_2 \leq \frac{1-\da\db}{1-\db}(1-\dbar), & 
\end{array} \right.
\end{equation}
where $\dbar = \eta (\da + \db) + (1-2\eta) \dm$.

\noindent \underline{\bf Comparing different scenarios:} Figure~\ref{Fig:PlanComparison} depicts the achievable region for the four scenarios we studied above.
As we can see, the dynamic RIS-user association enables the highest overall sum-rate.
The discussion above brings upon two questions: 
(i) Is the region in \eqref{Eq:AchBalanced} the best achievable region for the given RIS-user association? and (ii) Is a better RIS-user association possible?

Essentially, we have already answered the first question. Comparing the region described in \eqref{Eq:Theorem} and the one in \eqref{Eq:AchBalanced}, we observe that $\mathcal{R}_{\msf{dynamic}} \equiv \mathcal{C}^{\msf{out}}$,
meaning that for the specific parameters given in this section and the symmetric dynamic association, $\mathcal{C}_{\mathcal{A}} \equiv \mathcal{C}^{\msf{out}}$ from \eqref{Eq:Theorem} and the achievable region described in \eqref{Eq:AchBalanced} is indeed optimal. On the other hand, we cannot yet answer the second question.
There are multiple challenges that would be the subject of future work. 
First, the capacity region in general remains unknown and while the region described in \eqref{Eq:AchBalanced} is in fact the capacity region, we do not have the answer for all possible values of the parameters.
Second, the dynamic RIS-user association has a specific format (\emph{e.g.}, each RIS-user association last a linear function of the blocklength).
In other words, even if we characterize the capacity region of the dynamic association in general, it is not sufficient to fully answer the second question, suggesting stronger outer-bounds may be needed.

\section{Proof of Theorem~\ref{THM:RIS-BEC}}
\label{Section:Proofs}

In this section, we provide the proof of Theorem~\ref{THM:RIS-BEC}. 
The proof has some similarities to the one presented in recent results on broadcast channels such as~\cite{vahid2021erasure}, but with one main difference.
The RIS-user association, $\si$ is decided prior to the beginning of communications and as such, the statistical changes in the channel are known non-causally.
Thus, we focus more on the new aspects of the proof.
Moreover, we assume $\dm \leq \frac{\db(1-\da)}{1-\db}$, thus based on \eqref{Eq:beta}, $\beta = \frac{1-\da\db}{1-\db}$. 
The other scenario can be obtained similarly.
Finally, we will derive the first inequality in \eqref{Eq:Theorem} and the proof of the second one would be very similar.

For the problem defined in Section~\ref{Section:Setup}, fix $\si$ and suppose rate-tuple $\lp R_1, R_2 \rp$ is achievable. 
We have:
\begin{align}
&n \left( R_1 + \beta R_2 \right) = H(W_1) + \beta H(W_2) \nonumber \\
& \overset{(a)}= H(W_1|W_2, S^n, \si) + \beta H(W_2| S^n, \si) \nonumber \\
& \overset{(\mathrm{Fano})}\leq I(W_1;Y_1^n|W_2, S^n, \si) + \beta I(W_2;Y_2^n|S^n,\si) + n \upxi_n \nonumber \\
& = H(Y_1^n|W_2, S^n, \si) - \underbrace{H(Y_1^n|W_1,W_2,S^n,\si)}_{=~0} \nonumber \\
& \quad + \beta H(Y_2^n|S^n,\si) - \beta H(Y_2^n|W_2,S^n,\si) + n \upxi_n \nonumber \\
& \overset{(b)}\leq \beta H(Y_2^n|S^n,\si) + n \upxi_n \nonumber \\
& \overset{(c)}\leq n \beta \lp 1 - \eta_1 \dn - \eta_2 \dd - (1-\eta_1-\eta_2) \ds \rp + \upxi_n \nonumber \\
& \overset{(d)}\leq n \beta \left( 1 - \dbar_2 \right) + \upxi_n,
\end{align}
where from Fano's inequality, $\upxi_n \rightarrow 0$ as $n \rightarrow \infty$; $(a)$ follows from the independence of messages; $(b)$ comes from applying Lemma~\ref{Lemma:Leakage_BIC_No} below; $(c)$ is true since the entropy of a binary random variable is at most $1$ (or $\log_2(q)$ for packets in $\mathbb{F}_q$) and the channel statistics is governed by the dynamic RIS-user association. Dividing both sides by $n$ and let $n \rightarrow \infty$, we obtain the first inequality in \eqref{Eq:Theorem}.

\begin{lemma}
\label{Lemma:Leakage_BIC_No}
For the two-user double RIS-aided broadcast packet network as described in Section~\ref{Section:Setup} and for any encoding function satisfying~\eqref{eq_enc_function}, we have:
\begin{align}
\label{eq:lemma}
H\left( Y_1^n | W_2, S^n, \si \right) - \beta  H\left( Y_2^n | W_2, S^n, \si \right) \leq 0,
\end{align}
 
\end{lemma}


\begin{proof}
We have:
\begin{align}
&H\left( Y_2^n | W_2, S^n, \si \right) \nonumber \\
&\overset{(a)}= \sum_{t=1}^{\eta_1 n}{(1-\dn) H\left( X[t] | Y_2^{t-1},W_2, S^{n}, \si \right)} \nonumber \\
&+ \sum_{t=\eta_1 n+1}^{(\eta_1+\eta_2)n}{(1-\da) H\left( X[t] | Y_2^{t-1},W_2, S^{n}, \si \right)} \nonumber \\
&+ \sum_{t=(\eta_1+\eta_2) n+1}^{n}{(1-\dm) H\left( X[t] | Y_2^{t-1},W_2, S^{n}, \si \right)} \nonumber \\
&\geq \sum_{t=1}^{(\eta_1+\eta_2) n}{\frac{1}{\beta} H\left( Y_1[t], Y_2[t] | Y_1^{t-1}, Y_2^{t-1},W_2, S^{n}, \si \right)} \nonumber \\
&+ \sum_{t=(\eta_1+\eta_2) n+1}^{n}{\frac{H\left( Y_1[t], Y_2[t] | Y_1^{t-1}, Y_2^{t-1},W_2, S^{n}, \si \right)}{(1+\dm)}} \nonumber \\
&\overset{(b)}\geq \frac{1}{\beta} H\left( Y_1^n, Y_2^n | W_2, S^{n}, \si \right) \nonumber \\
&\overset{(c)}\geq \frac{1}{\beta} H\left( Y_1^n | W_2, S^{n}, \si \right),
\end{align}
where $(a)$ holds since $X[t]$ is independent of $S[t]$; $(b)$ holds since the omitted term is the product of a discrete entropy term with: 
\begin{align}
    \frac{1}{1+\dm} - \frac{1-\db}{1-\da\db},
\end{align}
which are both non-negative given that we assumed $\dm \leq \frac{\db(1-\da)}{1-\db}$; and $(c)$ follows from the non-negativity of the discrete entropy function.
\end{proof}


\section{Conclusion}
\label{Section:Conclusion}

We provided an information-theoretic approach to RIS-user association in the context of a two-user double RIS-aided broadcast packet network with goal of capacity maximization. 
We showed how a dynamic RIS-user association could enhance channel capacity.
We presented a new set of outer-bounds and an opportunistic linear protocol that achieves these bounds under non-trivial conditions.
We then discussed potential venues for further investigations.
In particular, for the dynamic association defined in this work, the capacity region still remains unknown and needs additional research.
Once that is answered, the best possible association may still take a different form.
Another potential direction is the extension of the results to interfering erasure channels~\cite{AlirezaBFICDelayed}, which is relevant to vehicular networks with multiple roadside units. 



\bibliographystyle{ieeetr}
\bibliography{bib_FBBudget}

\end{document}